\documentclass[11pt,a4paper]{article}
\pdfoutput=1 
\usepackage{a4wide,url}
\usepackage{graphicx,amssymb,amsmath,amsthm,xcolor}

\graphicspath{{figures/}{}}

\newtheorem{definition}{Definition}
\newtheorem{lemma}[definition]{Lemma}

\newtheorem{theorem}[definition]{Theorem}

\usepackage{bm}
\usepackage{enumitem}
\usepackage{mathabx}

\let\boxminusnew\boxminus
\let\boxplusnew\boxplus

\usepackage{MnSymbol}

\graphicspath{{figures/}{}}

\newcommand*{\boxangle}{%
\mathrel{
\vcenter{\offinterlineskip
\hbox{\scalebox{.76}[0.85]{$\righthalfcup$}}\vskip-.9ex\hbox{$\boxvoid$}}}}

\newcommand{\boxanglerota}{\rotatebox[origin=c]{90}{$\boxangle$}}
\newcommand{\boxanglerotb}{\rotatebox[origin=c]{270}{$\boxangle$}}

\newcommand{\boxtriangledown}{\rotatebox[origin=c]{180}{$\boxtriangleup$}}
\newcommand{\boxtriangleright}{\rotatebox[origin=c]{270}{$\boxtriangleup$}}
\newcommand{\boxtriangleleft}{\rotatebox[origin=c]{90}{$\boxtriangleup$}}

\newcommand{\boxvertnew}{\rotatebox[origin=c]{90}{$\boxminusnew$}}

\begin{document}

\title{Finding a largest empty convex subset in space is W[1]-hard}

\author{Panos Giannopoulos\thanks{Institut f{\"u}r Informatik, Universit{\"a}t Bayreuth, 
			  \allowbreak Universi\-t{\"a}tsstra{\ss}e, 30, D-95447 Bayreuth, Germany,
			  \textsf \{christian.knauer, panos.giannopoulos\}@uni-bayreuth.de}.
\footnote{Research supported by the German Science Foundation (DFG) under grant Kn~591/3-1.}
\and
Christian Knauer\footnotemark[1]
}

\index{Panos Giannopoulos}
\index{Christian Knauer}

\maketitle

\begin{abstract}
We consider the following problem: Given a point set in space find a largest subset that is in convex position and whose convex hull is empty. We show that the (decision version of the) problem is W[1]-hard.
\end{abstract}


\section{Introduction}
\label{sec:introduction}

\noindent
{\textbf {Problem definition.}} Let $P$ be a set of $n$ points in $\mathbb{R}^3$ and $k\in \mathbb{N}$. In the \textsc{Largest-Empty-Convex-Subset} problem we want to decide whether there is a set $Q\subset P$ of $k$ points in convex position whose convex hull does not contain any other point of $P$.  

\subsection{Results}

We show that \textsc{Largest-Empty-Convex-Subset} is W[1]-hard with respect to the solution size $k$, under the extra condition that the solution set is \emph{strictly} convex, i.e., the interior of the convex hull of any of its subsets is empty. This means that (under standard complexity-theoretic assumptions) the problem is not fixed-parameter tractable with respect to $k$, i.e., it does not admit an $O(f(k)\cdot n^c)$-time algorithm for any computable function $f$ and any constant $c$. See~\cite{FG06} for basic notions of parameterized complexity theory and~\cite{GKW08}  for survey of parameterized complexity results on geometric problems.

\subsection{Related work}

\textsc{Largest-Empty-Convex-Subset} has been shown to be NP-hard in last year's EuroCG \cite{KW12}. (In that paper, NP-hardness has been shown also for the more general version where the emptiness condition is dropped.) Several interesting questions were also raised such as whether the problem is fixed-parameter tractable with respect to the solution size and whether it admits a polynomial $o((\log n)/n)$-approximation algorithm. Here, we give a negative answer to the first question.
Note that in the plane, the problem is solvable in polynomial time; see, for example, \cite{AR85},~\cite{DEO90}.

From a combinatorial point of view, there is a long history of results starting with the famous Erd\"os-Szekeres theorem \cite{ES35}, which states that for every $k$ there is a number $n_k$ such that every planar set of $n_k$ points in general position contains $k$ points in convex position. Horton \cite{Ho83} showed that this is not true when the emptiness condition is imposed: There are arbitrarily large sets that do not contain empty $7$-gons. Results of this type exist also for higher dimensions; see~\cite{MS00}.


\section{Reduction}
\label{sec:reduction}

We show that \textsc{Largest-Empty-Convex-Subset} is W[1]-hard by an \emph{fpt}-reduction from the W[1]-hard $k$-\textsc{Clique} problem \cite{FG06}: Given a
graph $G([n], E)$ and $k\in \mathbb{N}$, decide whether $G$ contains a clique of size $k$.

\begin{figure}
\centering
	\includegraphics[width=.47\textwidth]{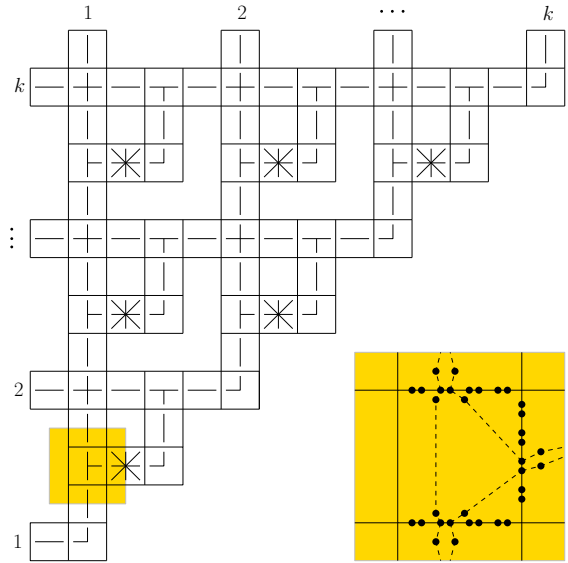}
	\caption{High level schematic of the construction. The zoomed-in area shows points shared among gadgets on their common boundaries and some pairs of points inside each gadget that take part in a choice of an empty convex set, (partially) marked with dashed segments. }
	\label{fig:high_level_construction}
\end{figure}

\subsection{High level description}

We begin with a high-level description of the construction, see Fig.~\ref{fig:high_level_construction}. Initially, the construction will lie on the plane; later on, it will be lifted to the elliptic paraboloid with a (more or less) standard transform. The construction is organized as the upper diagonal part of a grid with $k$ rows and $k$ columns. The $i$th row and column represent a choice for the $i$th vertex of a clique in $G$ and are made of $4i-2$ and $4(k-i+1)-2$ gadgets respectively. There are $n$ choices and each choice is represented by a collection of empty convex subsets of points -- one subset with a constant number of points from each gadget.

Each gadget consists of $\Theta(n)$ points within a rectangular region, which are organized in sets (of pairs) of collinear points. There is a constant number of such sets and, since we are looking for strictly convex subsets, only one pair of consecutive points per set can be chosen at any time. Certain choices are rendered invalid by additional points. Neighboring gadgets share the points on their common rectangle edge, see the zoomed-in area in Fig~\ref{fig:high_level_construction}. Through these common points, the choice of subsets is made consistent among the gadgets. In particular, the choice in the $i$th row is made consistent with the choice in the $i$th column via the `diagonal', $\boxangle$ gadget in their intersection corner; consistency here means that they both correspond to the same choice of a vertex of $G$. On the other hand, in the intersection of the $i$th column with the $j$th row, for every $j\geq i+1$, there is a `cross', $\boxplusnew$ gadget, which ensures that the choice in the column is propagated independently of the choice in the row and vice versa. 
Finally, the $j$th column is `connected' to the $i$th row, for every $j+1\leq i\leq k$, by three additional gadgets. One of them, the `star', $\boxasterisk$ gadget, encodes graph $G$, i.e., it allows only for combinations of choices (in the column and the row) that are consistent with the edges in $G$.

Locally, every valid subset from a gadget consists of points that are in strictly convex position and whose convex hull is empty. By lifting the whole construction to the paraboloid appropriately, we make sure that this property is true globally, i.e., for any set constructed from the local choices in a consistant manner. 

In total, the construction consists of a set $P$ of $\Theta(k^2n^2)$ points  in $\mathbb{R}^3$, such that there exists a set $Q\subset P$ with the desired property and $|Q| = f(k)$, for some function $f(k)\in \Theta(k^2)$, if and only if $G$ has a clique of size $k$.

\subsection{Gadgets}
\label{sec:gadgets}

There are five different types of gadgets, and each type has a specific function, which is explained below.

\medskip
\noindent
{$\boxminusnew$ \textbf {gadget.}} This gadget propagates a choice of pairs of points horizontally, see Fig.~\ref{fig:horizontal_gadget}. 
\begin{figure}[h]
\centering
	\includegraphics[width=.7\textwidth]{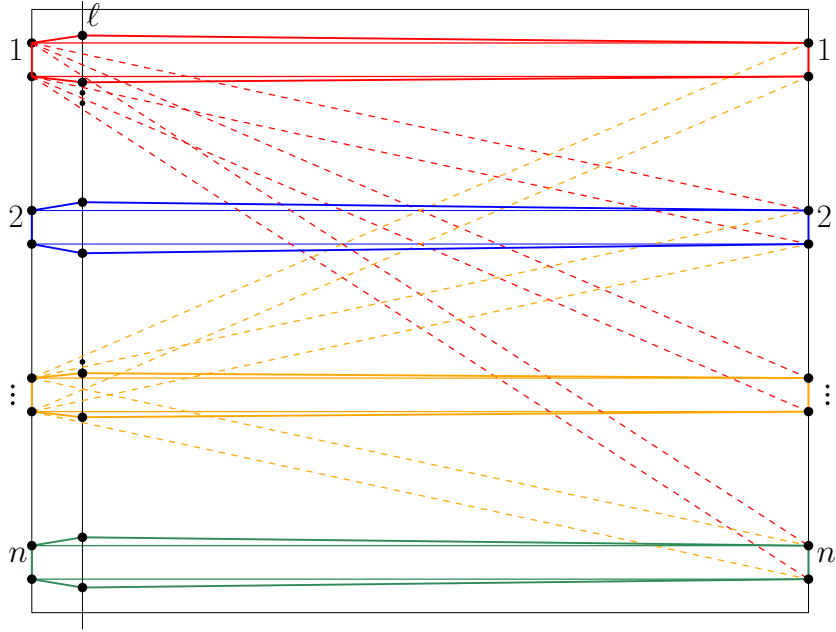}
	\caption{The $\boxminusnew$ gadget. Dashed parallelograms represent choices cancelled by points on $\ell$. Parallelograms in full are not cancelled. There are $n$ choices of convex empty $6$-gons.}
	\label{fig:horizontal_gadget}
\end{figure}
It has $n$ pairs of points $L_i$, $R_i$ on the left and right side of its rectangle respectively, which correspond to the vertices of $G$; pairs corresponding to the same vertex are aligned horizontally. It also has $n(n-1)+2$ points inside the rectangle on a vertical line $\ell$ as follows. For every two pairs $L_i$ and $R_j$, with $i\neq j$, a point is placed such that it is inside the parallelogram $L_iR_j$ formed by the pairs but outside the parallelogram formed by any other two pairs. Effectively, this point cancels a choice of pairs that correspond to different vertices of $G$. One point is placed above $L_1R_1$ on $\ell$ and close to its boundary; similarly one point is placed below $L_nR_n$.  Due to the strict convexity condition, at most one pair per rectangle side and at most one pair on $\ell$ can be chosen. Thus, there are $n$ maximum size empty convex subsets. Each subset contains six points and is formed by three pairs: $L_i$, $R_i$, for some $i$, and the pair of points on $\ell$ that are closest to the parallelogram $L_iR_i$; this latter pair is formed by the point that cancels $L_iR_{i-1}$ and the point that cancels $L_iR_{i+1}$. The $\boxvertnew$ gadget is just a $90^{\circ}$-rotated copy of the $\boxminusnew$ gadget.

\medskip
\noindent
{$\boxangle$ \textbf{gadget.}} This gadget has basically the same structure as the $\boxminusnew$ gadget but propagates information diagonally. For completeness, it is shown in Fig \ref{fig:diagonal_simple_gadget}.
\begin{figure}[h]
\centering
	\includegraphics[width=.7\textwidth]{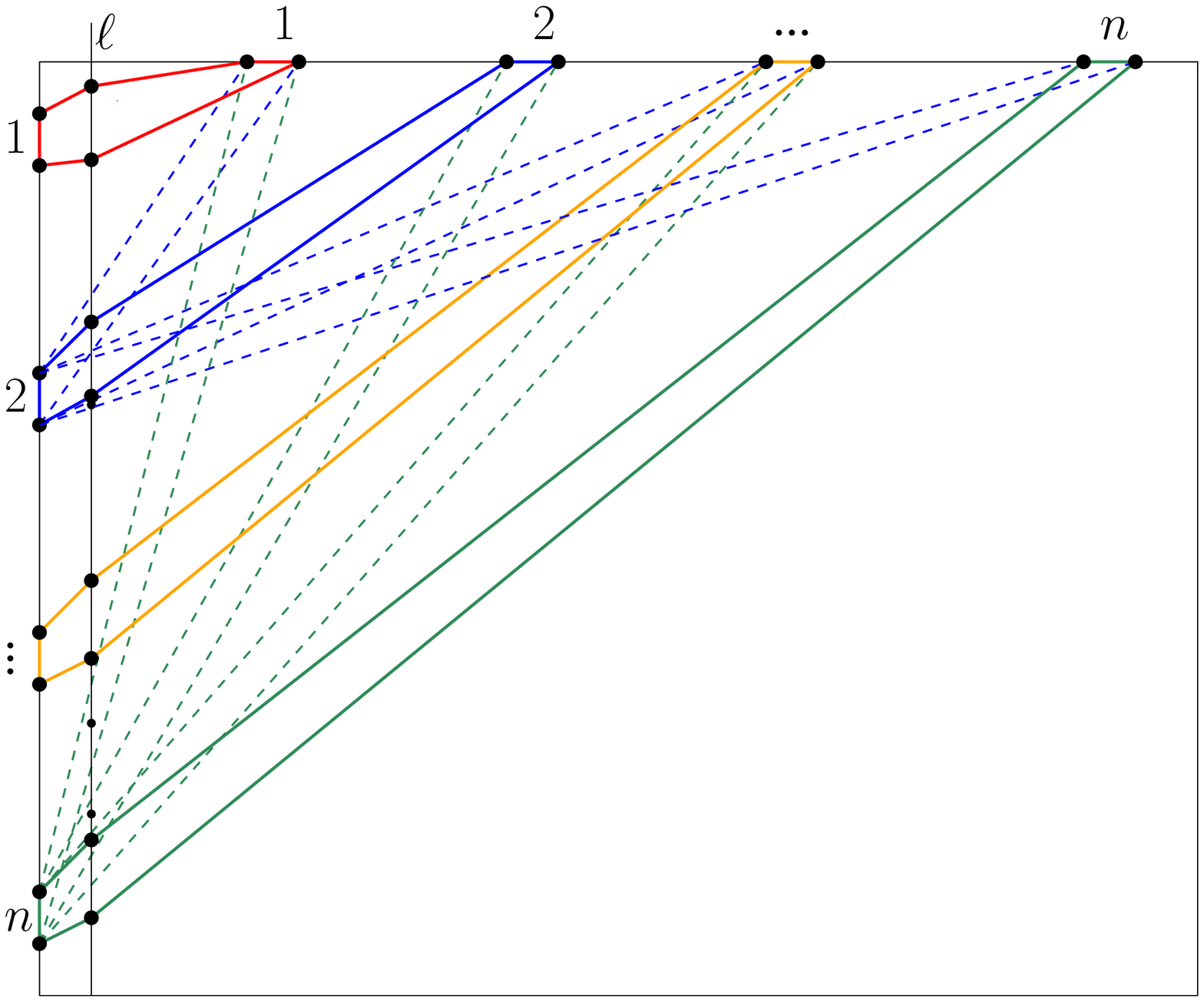}
	\caption[]{The $\boxangle$ gadget.}
	\label{fig:diagonal_simple_gadget}
\end{figure}

\medskip
\noindent
{$\boxtop$ \textbf{gadget.}} This gadget propagates information both horizontally and diagonally, see Fig.~\ref{fig:T_gadget}. It has $n$ pairs of points $L_i$, $R_i$, and $B_i$, on the left, right, and bottom side of its rectangle respectively. As before, an $i$th pair corresponds to the $i$th vertex of $G$. There will be only $n$ valid choices and three pairs per choice, namely, the $i$th pair from each side. This is enforced by the strict convexity condition and by placing, for every $i$, four additional points on two vertical lines $\ell$ and $\ell'$ inside the rectangle. These points are placed \emph{outside} the convex $6$-gon $L_iR_iB_i$ that is formed by the points in the corresponding pairs. (The points are also outside every other $6$-gon for $j\neq i$.) See the example for $i=2$ in Fig.~\ref{fig:T_gadget}. 

More specifically, looking at the gadget from top to bottom and from left to right, a point is placed on the intersection of $\ell$ and the line through the second point of $L_i$ and the second point of $B_{i-1}$; for the case of $i=1$, the point is placed just below the $6$-gon. A second point is placed on $\ell$ just above the $6$-gon. At the right side, two points are placed on $\ell'$ as follows. One point is placed on the intersection of $\ell'$ and the line through the second point of $L_i$ and the first point of $R_{i+1}$; for $i=1$, the point is placed just above the $6$-gon. A second point is placed on the intersection of $\ell'$ and the line through the second point of $R_i$ and the first point of $B_{i+1}$; for $i=n$, the point is placed just below the $6$-gon.

There are $n$ maximum size empty convex subsets with 10 points each. A subset is formed by the points in the pairs $L_i$, $R_i$, $B_i$, and the four points on $\ell$ and $\ell'$ that are closest to the corresponding $6$-gon. 

(Note that gadgets $\boxright$ and $\boxleft$ (used later on) are just rotated copies of gadget $\boxtop$.)
\begin{figure}
\centering
	\includegraphics[width=.7\textwidth]{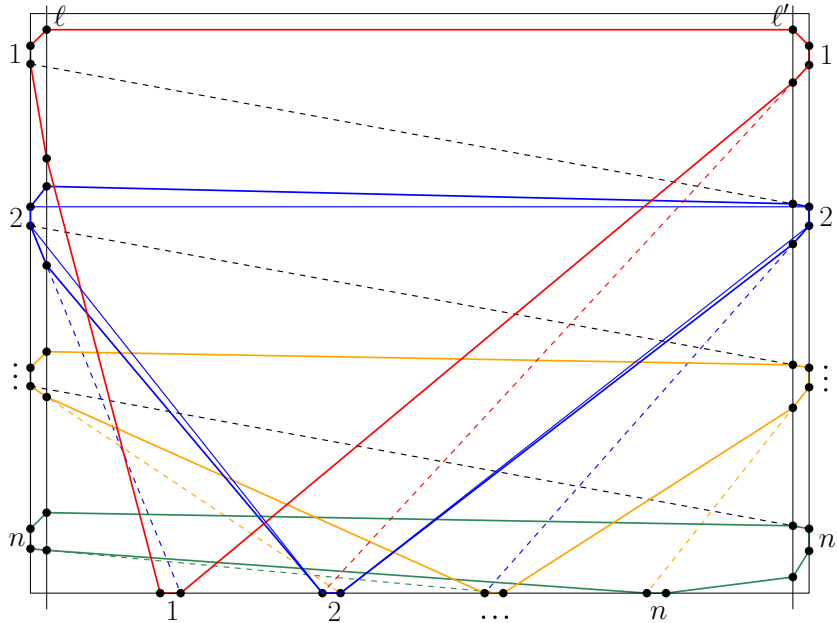}
	\caption{The $\boxtop$ gadget. There are $n$ choices of convex empty $10$-gons. An empty convex $6$-gon, as described in the text, is shown only for $i=2$.}
	\label{fig:T_gadget}
\end{figure}

\medskip
\noindent
{$\boxplusnew$ \textbf{gadget.}} This gadget consists of eight subgadgets, which are very similar to the ones we have already described above. See Fig.~\ref{fig:high_level_cross_gadget}.
\begin{figure}[h]
\centering
	\includegraphics[width=.5\textwidth]{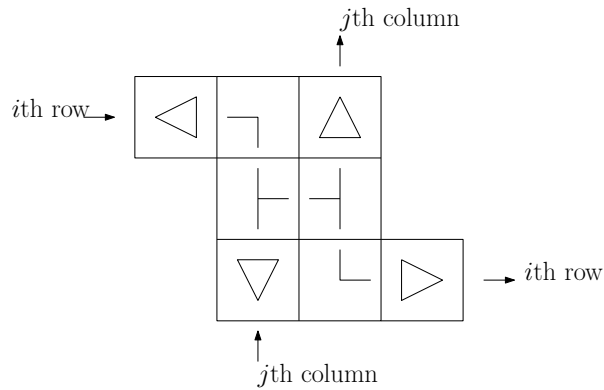}
	\caption{The $\boxplusnew$ gadget: a high level schematic.}
	\label{fig:high_level_cross_gadget}
\end{figure}
It can be thought of as having two inputs (at the upper and lower left corner) and two outputs (at the upper and lower right corner). It propagates the inputs (choices) independently from each other, one horizontally and one vertically. The input subgadgets $\boxtriangleleft$ and $\boxtriangledown$ are respectively connected (through their left and bottom rectangle sides) to the $\boxminusnew$ and $\boxvertnew$ gadgets of the $i$th row and $j$th column of the global construction (Fig.~\ref{fig:high_level_construction}). The output subgadgets $\boxtriangleright$ and $\boxtriangleup$ are similarly connected to a $\boxminusnew$ and $\boxvertnew$ gadget. Note that the subgadgets $\boxright$ and $\boxleft$ in the middle of the $\boxplusnew$ gadget (Fig.~\ref{fig:high_level_cross_gadget}) as well as the subgadgets $\boxanglerota$ and $\boxanglerotb$ are \emph{not} connected directly to any row or column of the global construction. Roughly speaking, the $\boxplusnew$ gadget has the following function: it multiplexes the two inputs, then it mirrors them (vertically and horizontally), and then demultiplexes them. 
Next, we describe the subgadgets in some more detail.

\begin{figure}[h]
\centering
	\includegraphics[width=.5\textwidth]{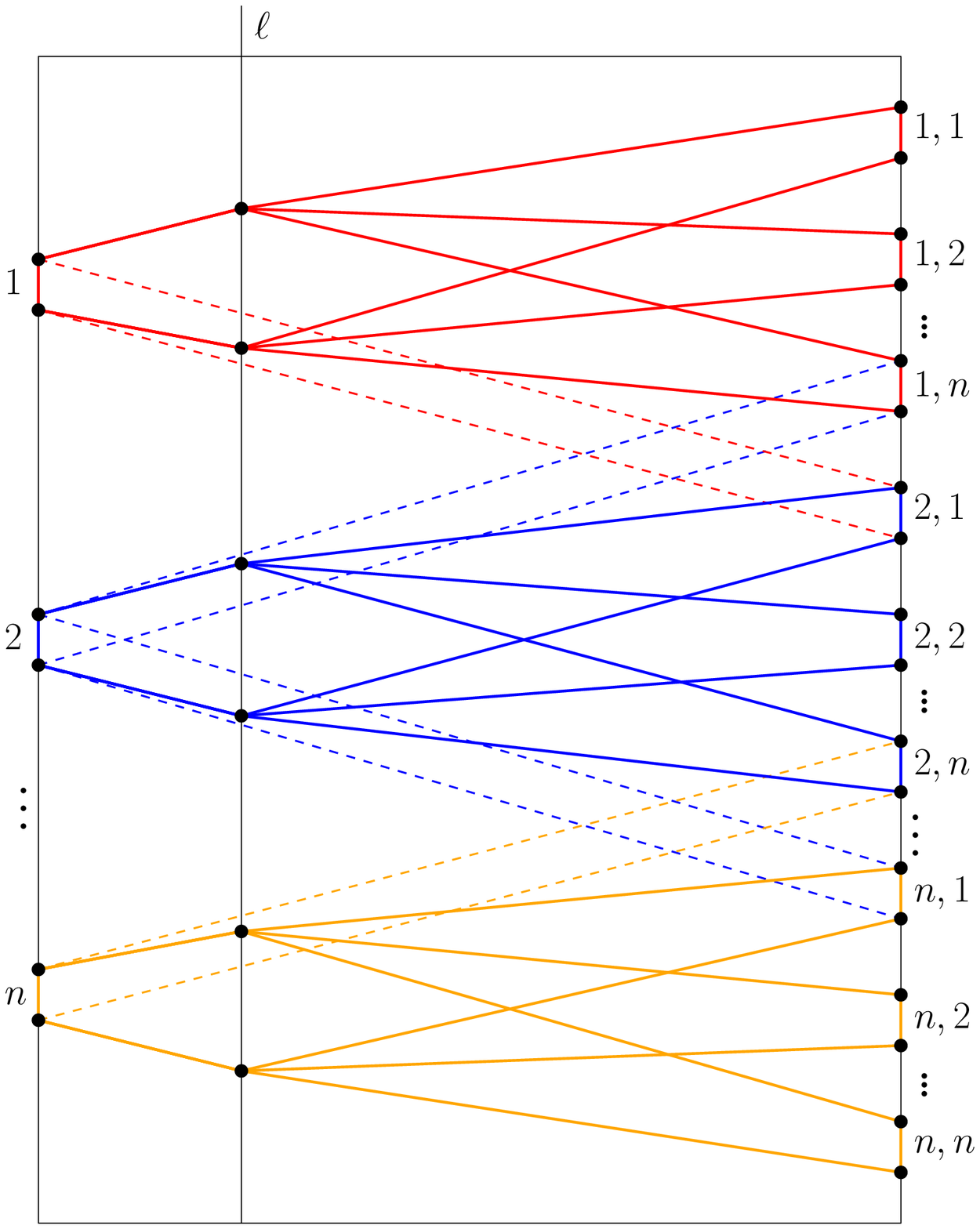}
	\caption[]{The $\boxtriangleleft$ subgadget.}
	\label{fig:horizontal_for_cross_gadget}
\end{figure}

The $\boxtriangleleft$ subgadget is shown in Fig.~\ref{fig:horizontal_for_cross_gadget}. It is similar to the previously described $\boxminusnew$ gadget in Fig.~\ref{fig:horizontal_gadget}. It has again $n$ pairs of points $L_i$ on the left side. The difference now is that there are $n^2$ pairs of points $R_{i,j}$, $1\leq i, j\leq n$, on the right side of the gadget. The second index $j$ basically encodes the choice coming from the input subgadget $\boxtriangledown$ at the lower left corner, which is communicated through the $\boxright$ and $\boxanglerota$ subgadgets inbetween. Only the $n^2$ pairs $L_i$ and $R_{i,j}$ constitute valid choices. The rest are cancelled by additional points as usually. Together with pairs of canceling points (which are also chosen as before) there are exactly $n^2$ empty convex $6$-gons, and these are of maximum size.

The input $\boxtriangledown$ subgadget propagates information vertically and is defined similarly to the $\boxtriangleleft$ subgadget. It has $n$ pairs of points $B_j$ on the bottom side and $n^2$ pairs $T_{ij}$ on the top side. The difference now is that only the $n^2$ pairs $B_j$ and $T_{i,j}$ are valid.


The output subgadgets $\boxtriangleright$ and $\boxtriangleup$ are just mirrored images of their input counterparts. The $\boxanglerota$ and $\boxanglerotb$ subgadgets are constructed in the same way as the $\boxangle$ gadget in Fig.~\ref{fig:diagonal_simple_gadget} but have $n^2$ valid $6$-gons, while the $\boxright$ and $\boxleft$ subgadgets are constructed in the same way as the $\boxtop$ gadget in Fig.~\ref{fig:T_gadget} and have $n^2$ valid $10$-gons. 

\medskip
\noindent
{$\boxasterisk$ \textbf{gadget.}} This gadget encodes the edges of the input graph $G$. See Fig.~\ref{fig:star_gadget}. It is similar to gadget $\boxminusnew$ (Fig.~\ref{fig:horizontal_gadget}) and allows only combinations of pairs that correspond to edges of the graph: for every \emph{non}-edge $ij$ of $G$, a point is placed inside the parallelogram $L_iR_j$. 

\begin{figure}[h]
\centering
	\includegraphics[width=.65\textwidth]{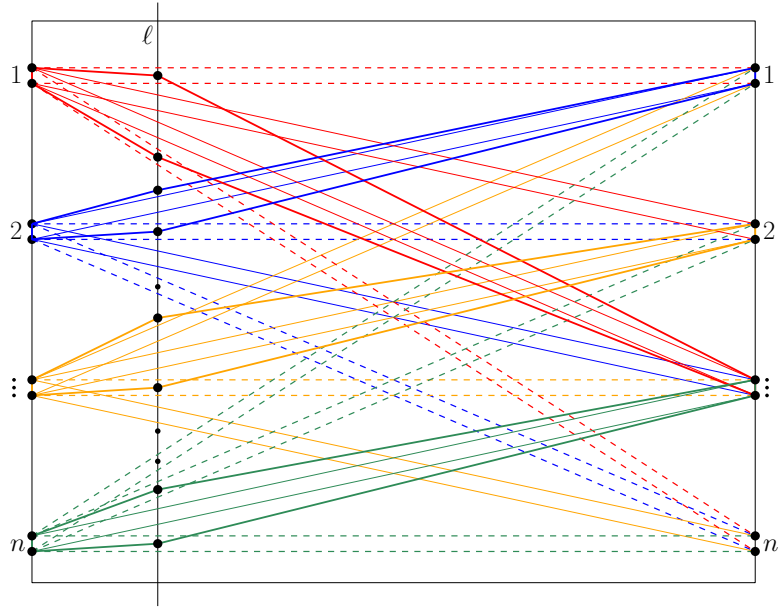}
	\caption{The $\boxasterisk$ gadget. Several examples of cancelled choices (in dashed) and of empty convex $6$-gons (in bold) are shown.}
	\label{fig:star_gadget}
\end{figure}

\subsection{Lifting to $\mathbb{R}^3$}
Every corner of a gadget rectangle is lifted to the paraboloid with the map $(x, y) \mapsto (x ,y , x^2+y^2)$. See Fig~\ref{fig:lifting}. The images of the corners of each rectangle lie on one distinct plane (since the corners lie on a circle). The points in a gadget are projected orthogonally on the corresponding plane. This is an affine map and thus colinearity and convexity within a gadget is preserved. Each gadget now lies on a distinct facet (a parallelogram) of a convex polyhedron.

\begin{figure}[h]
\centering
	\includegraphics[width=.5\textwidth]{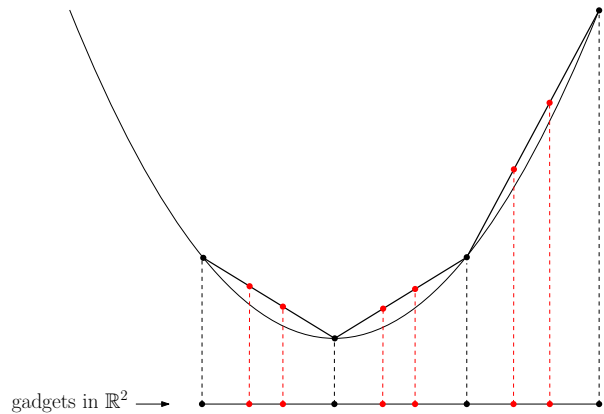}
	\caption{A $2$-dimensional analogue of the lifting procedure.}
	\label{fig:lifting}
\end{figure}

\subsection{Correctness}
The total number of (sub)gadgets of each type together with the size of a largest valid (i.e., empty and convex) subset in a gadget of the type is
\begin{itemize}[leftmargin=2.2cm]
\setlength{\itemsep}{-\parsep}
\item[$\boxminusnew$] : $k^2$, $6$;
\item[$\boxvertnew$] : $k(3k-1)/2$, $6$;
\item[$\boxangle$, $\boxanglerota$, $\boxanglerotb$] : $k(3k-1)/2$, $6$;
\item[$\boxtop$, $\boxright$, $\boxleft$] : $2k(k-1)$, $10$;
\item[$\boxtriangleleft$, $\boxtriangleright$, $\boxtriangleup$, $\boxtriangledown$] : $2k(k-1)$, $6$;
\item[$\boxasterisk$] : $k(k-1)/2$, $6$.
\end{itemize}

A global valid subset is formed by locally choosing one valid subset from every gadget in a consistent manner. When a largest locally possible subset (as given above) can be chosen, the global subset has size $k(35k-23)$. As we will now prove, such a global subset corresponds to a $k$-size clique of $G$. Let $P$ be the set of all the points in our construction.

\begin{lemma}
There exists an empty convex subset of $P$ with $k(35k-23)$ points if and only if $G$ has a clique of size $k$.
\end{lemma}
\begin{proof}
Suppose there exists a global valid subset of size $k(35k-23)$. Then, a largest locally possible valid subset must be chosen from every gadget. Consider such a choice of subsets and let $v_i$ be the vertex of $G$ corresponding to the choice from the leftmost gadget of the $i$th grid row. By construction, the subset corresponding to the same vertex $v_i$ must be chosen from every other gadget in this row as well as every gadget in the $i$th column. Consider the $j$th row, for some $j\neq i$. Through the $\boxasterisk$ gadget that connects the $i$th column to the $j$th row, when $i+1\leq j$, or the $j$th column to the $ith$ row, when $j<i$, the subset chosen from the $j$th row must correspond to a vertex $v_j$ such that $v_iv_j$ is an edge of $G$. Hence $\{v_1,\ldots v_k\}$ is a clique in $G$. The converse is obvious.
\end{proof}

Therefore, we have shown the following

\begin{theorem}
\textsc{Largest-Empty-Convex-Subset} is W[1]-hard, under the condition of strict convexity.
\end{theorem}

\bibliographystyle{abbrv}

\end{document}